\documentclass[11pt,twocolumn]{article}

\usepackage{amsmath,amssymb,amsthm}
\usepackage{bm}
\usepackage{booktabs}
\usepackage{tabularx}
\usepackage{graphicx}
\usepackage[hidelinks]{hyperref}
\usepackage[margin=1in]{geometry}
\usepackage{enumitem}
\usepackage{natbib}
\usepackage{xcolor}
\usepackage{pgfplots}
\usepackage{placeins}
\pgfplotsset{compat=1.18}
\usepgfplotslibrary{fillbetween}

\newtheorem{theorem}{Theorem}
\newtheorem{definition}{Definition}
\newtheorem{lemma}{Lemma}

\title{The Preservation Tradeoff: A Thermodynamic Bound in the Diminishing-Returns Regime}
\author{Amadeus Brandes\\
Independent Researcher, Germany\\
\texttt{brandesamadeus@gmail.com}}
\date{}

\begin{document}

\maketitle

\begin{abstract}
Thermodynamic systems that preserve information against thermal fluctuations face a tradeoff distinct from transmission (Shannon) or erasure (Landauer). We formalize the preservation problem by defining the preservation stiffness $\mathcal{S}_\kappa$, a response function analogous to magnetic susceptibility, and derive the Stiffness-Odds Identity: at optimal allocation, the stiffness equals the ratio of payload to maintenance capacity. This identity is the paper's central contribution. It reduces optimal preservation to a single measurable response variable and provides a substrate-agnostic diagnostic for thermodynamic efficiency---applicable wherever maintenance competes with payload, regardless of whether the underlying substrate is biochemical, electronic, or algorithmic.

For all systems in the diminishing-returns regime, we prove the unconditional bound $\kappa^* < 0.50$.

For the subclass exhibiting smooth saturation with rate parameter $a \in [2,3]$---an empirically characterized efficiency frontier, not a universal constant---the optimum is further constrained to the 30--50\% band. We motivate this functional form from two independent physical principles: Shannon error exponents and thermodynamic dissipation bounds. We then illustrate consistency with representative operating points from kinetic proofreading in \textit{E.~coli} and protocol overhead in TCP/IP networks, and specify conditions under which the framework is falsifiable.
\end{abstract}

\section{Introduction}

The thermodynamics of information is grounded in two well-defined limits. Shannon's noisy channel coding theorem bounds the rate of information transmission by the channel capacity $C$~\cite{Shannon1948}. Landauer's principle bounds the cost of information erasure by $k_B T \ln 2$~\cite{Landauer1961}.

However, physical systems often face a third distinct challenge: \emph{preservation}. From DNA repair mechanisms to error-correcting memory in computation, systems must actively suppress entropy generation to maintain a target macrostate over time. Unlike transmission (moving bits) or erasure (resetting bits), preservation involves the continuous maintenance of state fidelity against a thermal bath. While processes like kinetic proofreading occur during the transmission of genetic information, we treat the excess dissipation driving the discrimination step specifically as the maintenance overhead required to preserve fidelity.

This problem has been treated in specific contexts---reliability engineering, kinetic proofreading~\cite{Hopfield1974}, and finite-blocklength coding theory~\cite{Gallager1965,Polyanskiy2010}---but always in model-specific, dimensional terms. Unlike prior work analyzing specific mechanisms, we derive a substrate-agnostic constraint that links maintenance overhead directly to a response function $\mathcal{S}_\kappa$, yielding a bound comparable in form to capacity or Landauer's limit.

Related perspectives include stochastic thermodynamics of information processing~\cite{Seifert2012,Parrondo2015} and the finite-blocklength approach to channel coding~\cite{Polyanskiy2010}, which bound error probability as a function of redundancy. Our contribution is to recast the preservation problem in the language of fluctuation-response theory, identifying the stiffness $\mathcal{S}_\kappa$ as the natural control variable.

Our central result is the Stiffness-Odds Identity (Eq.~\ref{eq:identity}), which shows that at optimal allocation the preservation stiffness $\mathcal{S}_\kappa$ must equal the resource odds $(1-\kappa^*)/\kappa^*$. This gives a structural criterion for preservation: it converts the problem of predicting optimal maintenance into the problem of measuring a single response function, independently of the physical substrate. From this identity we prove an unconditional bound: for any system in the diminishing-returns regime, $\kappa^* < 0.50$. For the subclass with smooth saturation near the efficiency frontier ($a \in [2,3]$), the bound sharpens to a 30--50\% band. We motivate this regime from two independent physical routes---information-theoretic constraints on error exponents and thermodynamic constraints on dissipation---whose convergence on the same exponential-saturation form (derived in Section~\ref{sec:physical_derivation}) gives the regime a physical basis rather than treating it as a purely phenomenological fit. The inclusion of a non-biochemical system (TCP/IP) alongside molecular proofreading is deliberate: it probes whether the band reflects thermodynamic dissipation specifically or the generic allocation structure $(1-\kappa)\Delta R(\kappa)$ under diminishing returns, with $\Delta R$ denoting reliability gain above any nonzero baseline.

\section{Formalism: The Preservation Stiffness}

Consider a system with total resource capacity normalized to unity. This capacity is partitioned into payload $(1-\kappa)$ and maintenance $\kappa$. The system's effective throughput is defined as
\begin{equation}
T_{\mathrm{eff}}(\kappa) = (1-\kappa)R(\kappa),
\label{eq:throughput}
\end{equation}
where $R(\kappa)$ is the reliability (probability of correctness) as a function of maintenance investment.

To establish the generality of this framework, we map the variables to their physical counterparts in Table~\ref{tab:rosetta}. This mapping reveals that preservation is a universal thermodynamic class, distinct from the specific substrate (biological, electronic, or algorithmic).

\begin{table*}[t]
\centering
\caption{The Preservation Rosetta Stone: Mapping thermodynamic variables across substrates. The stiffness $\mathcal{S}_\kappa$ serves as the universal response function linking these domains.}
\label{tab:rosetta}
\renewcommand{\arraystretch}{1.4}
\begin{tabular}{l c l l}
\toprule
\textbf{Variable} & \textbf{Symbol} & \textbf{Thermodynamics} & \textbf{Protocol / Biology} \\
\midrule
Capacity & $1$ & Free Energy ($F$) & Bandwidth / GTP Budget \\
Maintenance & $\kappa$ & Dissipated Work ($W_{\text{diss}}$) & Overhead / Proofreading \\
Reliability & $R(\kappa)$ & State Fidelity & Goodput / Accuracy \\
Efficiency & $a$ & Coupling ($\Delta E / k_B T$) & Signal-to-Noise Ratio \\
Response & $\mathcal{S}_\kappa$ & Stiffness (Inverse $\chi$) & Saturation metric \\
\bottomrule
\end{tabular}
\end{table*}

To characterize the response of the system to maintenance injection, we define the \emph{preservation stiffness}
\begin{equation}
\mathcal{S}_\kappa \equiv \frac{R(\kappa)}{\kappa R'(\kappa)} = \left(\frac{\partial \ln R}{\partial \ln \kappa}\right)^{-1}.
\label{eq:stiffness}
\end{equation}
This quantity is the inverse of the elasticity of reliability. In thermodynamic terms, it measures the stiffness of the error-suppression mechanism. A low $\mathcal{S}_\kappa$ implies that reliability is highly sensitive to marginal maintenance investment (soft mode), while a high $\mathcal{S}_\kappa$ implies saturation (stiff mode).

Assume $R$ is differentiable on $(0,1)$ and that the optimum occurs at an interior point $\kappa^* \in (0,1)$; corner optima ($\kappa^* \in \{0,1\}$) correspond to degenerate regimes (no maintenance or no payload) and are excluded from the stiffness condition. Maximizing Eq.~\eqref{eq:throughput} yields the first-order condition $-R(\kappa^*) + (1-\kappa^*)R'(\kappa^*) = 0$. Rearranging and substituting the definition of $\mathcal{S}_\kappa$, we obtain the \textbf{Stiffness-Odds Identity}:
\begin{equation}
\boxed{\mathcal{S}_\kappa(\kappa^*) = \frac{1-\kappa^*}{\kappa^*}}
\label{eq:identity}
\end{equation}

\noindent\textbf{This identity is the central result of this paper.} All subsequent bounds on $\kappa^*$ follow from Eq.~\eqref{eq:identity} combined with constraints on $\mathcal{S}_\kappa$. The identity converts the problem of predicting optimal allocation into the problem of bounding a single response function. The right-hand side is the \emph{resource odds}---the ratio of payload capacity to maintenance capacity. At optimal efficiency, the system's stiffness must match this resource leverage.

This identity provides a diagnostic for thermodynamic efficiency. If $\mathcal{S}_\kappa > (1-\kappa)/\kappa$, the system is in a \emph{stiff mode}---over-investing in saturated maintenance and wasting payload capacity. If $\mathcal{S}_\kappa < (1-\kappa)/\kappa$, the system is in a \emph{soft mode}---under-investing despite high responsiveness, resulting in suboptimal throughput due to excessive error rates. Equilibrium occurs when stiffness matches resource odds.

For systems with nonzero baseline reliability ($R(0) \neq 0$), the framework applies to the normalized reliability gain above baseline; we introduce the appropriate normalization in Section~\ref{sec:validation} when analyzing TCP.

\section{Physical Derivation of the Diminishing-Returns Regime}
\label{sec:physical_derivation}

The critical physical question is the functional form of $R(\kappa)$. We now show that the exponential-saturation form $R(\kappa) = 1 - e^{-g(\kappa)}$ with concave $g$ emerges from two independent physical principles. This dual derivation---from information theory and thermodynamics---provides the physical motivation for the regime studied here: the same functional constraint arises from entirely different physical considerations.

\subsection{Information-Theoretic Route}
\label{sec:shannon_route}

Consider a system that must maintain state fidelity over a noisy channel with finite capacity $C$. The fundamental constraint is Shannon's noisy channel coding theorem: reliable communication requires coding rate $r < C$~\cite{Shannon1948}.

For finite blocklength $n$, Gallager's error exponent bounds give the probability of decoding error~\cite{Gallager1965}:
\begin{equation}
P_{\mathrm{error}} \sim \exp(-n \cdot E(r)),
\label{eq:gallager}
\end{equation}
where $E(r)$ is the reliability function (error exponent) of the channel, which depends on the coding rate $r$ and satisfies $E(r) > 0$ for $r < C$ and $E(C) = 0$.

Now interpret $n$ as the redundancy devoted to error protection, proportional to maintenance investment $\kappa$. Writing $n = \lambda \kappa$ for some constant $\lambda$ (the redundancy per unit maintenance), we have:
\begin{equation}
\varepsilon(\kappa) = 1 - R(\kappa) \sim \exp(-\lambda \kappa \cdot E(r)).
\label{eq:error_kappa}
\end{equation}

\textbf{Key property:} The reliability function $E(r)$ is concave in the coding rate $r$ and satisfies $E(r) \to 0$ as $r \to C$. Finite-blocklength analysis~\cite{Polyanskiy2010} refines this picture: for systems operating below capacity, the achievable error probability decays at most exponentially in blocklength, with rate determined by $E(r)$.

Inverting, we obtain $R(\kappa) = 1 - \exp(-g(\kappa))$ where $g(\kappa) = \lambda \kappa \cdot E(r)$. For a fixed operating rate $r < C$, $E(r)$ is a positive constant and $g(\kappa)$ is \emph{linear} (hence concave) in $\kappa$. Nonlinear concave $g$ arises when the effective exponent per unit redundancy decreases with $\kappa$ due to coding/decoding constraints, feedback, or other saturation mechanisms; we capture this general behavior with the concave $g$ assumption in Definition 1.

\textbf{Key constraint from information theory:} The rate parameter $a = g'(0)$ depends on both the channel (through $E(r)$) and the coding efficiency (through $\lambda$). For any fixed channel and coding scheme operating below capacity, $a$ is finite and positive. The essential point is not a specific bound on $a$, but rather that the exponential-saturation form $R(\kappa) = 1 - e^{-g(\kappa)}$ with concave $g$ emerges from the structure of channel coding.

\subsection{Thermodynamic Route}
\label{sec:landauer_route}

Consider the same preservation problem from the perspective of thermodynamic dissipation. Maintaining reliability $1 - \varepsilon$ against thermal fluctuations requires the system to operate as a driven nonequilibrium steady state (NESS), continuously exporting entropy to sustain a low-error configuration. As recently derived for molecular templating networks, maintaining such non-equilibrium low-entropy states imposes a fundamental free-energy cost~\cite{Ouldridge2025}.

The core constraint comes from the thermodynamics of information~\cite{Parrondo2015}. To maintain a target state against thermal drift, a system must continuously:
\begin{enumerate}[leftmargin=*,itemsep=2pt]
    \item \emph{Detect} deviations from the correct state (acquire information about errors);
    \item \emph{Correct} those deviations (dissipate free energy to restore the target).
\end{enumerate}
Both steps require entropy production. Stochastic thermodynamics~\cite{Seifert2012} establishes that the information gained about a system's state is bounded by the dissipation:
\begin{equation}
I_{\mathrm{gained}} \leq \frac{W_{\mathrm{diss}}}{k_B T \ln 2}.
\label{eq:info_diss_bound}
\end{equation}
Since error correction requires gaining information about errors, the rate of error suppression is bounded by available dissipation.

By Landauer's principle~\cite{Landauer1961}, erasing one bit of information requires dissipation of at least $k_B T \ln 2$. Each corrected error thus requires entropy export proportional to the information content of the error. For a system maintaining error probability $\varepsilon$ in steady state, the required dissipation rate scales with the log-odds of correctness:
\begin{equation}
\dot{W}_{\mathrm{diss}} \propto \ln(1/\varepsilon).
\label{eq:diss_rate}
\end{equation}

If maintenance investment $\kappa$ translates to sustained dissipation $W = \gamma \kappa$ (where $\gamma$ is the available free energy per unit maintenance), inverting Eq.~\eqref{eq:diss_rate} yields the steady-state error probability:
\begin{equation}
\varepsilon \sim \exp\!\left(-\frac{\gamma \kappa}{k_B T}\right),
\label{eq:thermo_error}
\end{equation}
and thus the reliability:
\begin{equation}
R(\kappa) = 1 - \exp\!\left(-\frac{\gamma \kappa}{k_B T}\right).
\label{eq:thermo_reliability}
\end{equation}

This exponential-saturation form---the same structure derived from Shannon error exponents in Section~\ref{sec:shannon_route}---emerges from the fundamental bound on how efficiently dissipation can be converted to error suppression. In minimal proofreading and copying networks, this scaling is the natural steady-state form; here we use it as the canonical smooth-regime representation rather than as a claim about all NESS architectures. (The Crooks fluctuation theorem~\cite{Crooks1999} provides complementary intuition: it relates the ratio of forward to reverse transition probabilities to dissipated work. However, preservation is better understood as steady-state maintenance than as a sequence of forward/reverse transient trajectories.)

\textbf{Key constraint from thermodynamics:} The rate parameter $a = g'(0) = \gamma/(k_B T)$ is bounded above by the available free energy per unit maintenance:
\begin{equation}
a \leq \frac{\Delta G_{\mathrm{available}}}{k_B T \cdot \text{(maintenance unit)}}.
\label{eq:a_thermo_bound}
\end{equation}

\subsection{Convergence: The Efficiency Frontier as an Attractor}
\label{sec:convergence}

The information-theoretic and thermodynamic routes converge on the same functional form:
\begin{equation}
R(\kappa) = 1 - \exp(-g(\kappa)), \quad g \text{ concave, } g(0) = 0, g' > 0.
\label{eq:unified_form}
\end{equation}

This convergence is not coincidental. \textbf{Error correction \emph{is} entropy export.} Shannon's channel capacity and Landauer's erasure bound are manifestations of the same underlying constraint: the rate at which a system can reduce uncertainty (whether interpreted as decoding errors or thermodynamic fluctuations) is bounded by its information-processing or dissipation capacity.

The rate parameter $a = g'(0)$ is not a universal constant but a system-specific figure of merit representing the coupling efficiency between maintenance resources and error suppression. It admits dual interpretations:
\begin{itemize}[leftmargin=*,itemsep=2pt]
    \item \textbf{Information theory:} Channel capacity per unit redundancy
    \item \textbf{Thermodynamics:} Available free energy per unit maintenance (in units of $k_B T$)
\end{itemize}

While $a$ can theoretically take any positive value, systems subject to strong optimization pressures---whether evolutionary (biology) or engineering (protocol design)---exhibit a striking convergence to a narrow range. We identify $a \in [2, 3]$ as the \emph{Thermodynamic Efficiency Frontier} for molecular and information-processing machines that use standard chemical fuels (ATP) and operate at typical biophysical efficiencies (see Appendix~\ref{app:dissipation} for empirical characterization).

The clustering of the rate parameter in this range reflects a fundamental finite-time thermodynamic tradeoff. Systems with $a < 2$ are noise-dominated, where thermal fluctuations spontaneously reverse error correction, leading to instability. Conversely, systems with $a \gg 3$ approach the reversible limit (``frozen'' regime), achieving high accuracy but at vanishingly slow operation speeds. The interval $a \in [2, 3]$ represents the maximum-power regime---the optimal compromise between thermodynamic stability and operational speed.

\paragraph{The Conditional Prediction.} The bound $a \in [2, 3]$ is not derived from first principles as a universal constant; it represents an empirical characterization of well-adapted systems operating at finite speed. The theory does not assert that all systems must have $a \in [2, 3]$. A system with very strong coupling (large $\eta_{\mathrm{cpl}}$ or $W_{\mathrm{budget}}$) could exhibit $a \gg 3$, but such systems typically approach the reversible limit and operate slowly. Conversely, a system with weak coupling could exhibit $a < 2$, but would be noise-dominated and unstable. Instead, the theory makes a \emph{falsifiable conditional prediction}: Given a system that has been optimized for throughput and reliability to operate near this efficiency frontier ($a \in [2, 3]$), its optimal maintenance allocation $\kappa^*$ is constrained to the 30--50\% band.

The observation that \textit{E.~coli} kinetic proofreading ($\kappa_{\mathrm{obs}} \approx 0.37$), TCP bulk transfer ($\kappa_{\mathrm{obs}} \approx 0.35$), and non-enzymatic DNA ligation~\cite{Aoyanagi2025} all independently occupy the same maintenance band suggests that $\kappa \in [0.30, 0.50]$ is a robust conditional optimum across disparate optimization landscapes when diminishing returns and smooth saturation apply. In molecular systems, this is consistent with operation near the thermodynamic efficiency frontier ($a \in [2,3]$); in TCP, it arises primarily from distributed stability constraints (see Section~\ref{sec:validation}).

\paragraph{Physical Interpretation of $a$.}

\textbf{Biological systems.} At $T = 310$~K, ATP/GTP hydrolysis provides $\Delta\mu \sim 20\, k_B T$ per molecule. Kinetic proofreading expends 2--4 molecules per discrimination step~\cite{Hopfield1974,Johansson2012}, but only a fraction $\eta_{\mathrm{cpl}}$ of this energy couples to the discrimination bias (see Appendix~\ref{app:dissipation} for the normalization). For well-adapted systems, the effective rate parameter $a \approx 2$--$3$ reflects this coupling efficiency.

\textbf{Network protocols.} TCP/IP provides a non-thermodynamic example where convergence to the preservation band arises from a different mechanism: distributed stability constraints rather than dissipation bounds. The observed $\kappa_{\mathrm{obs}} \approx 0.30$--$0.40$ in bulk TCP reflects overhead from acknowledgments, retransmissions, and---most significantly---capacity left unused by congestion control algorithms that maintain stability margin. Unlike molecular systems where $a \in [2,3]$ characterizes energetic coupling, TCP's convergence to the preservation band is better understood as a macroscopic stability phenomenon: in practice, protocols that attempt to operate below $\kappa \approx 0.30$ suffer instability under congestion, while those exceeding $\kappa \approx 0.50$ underutilize capacity. The preservation band is thus a mathematical consequence of the allocation objective $(1-\kappa)R(\kappa)$ combined with diminishing returns; the physical mechanisms generating those diminishing returns differ across substrates, but the optimal allocation constraint does not. (A stylized ARQ retransmission model is presented in Section~\ref{sec:validation} for comparison, but should not be conflated with the macroscopic TCP operating point.)

\textbf{Information-theoretic systems.} For binary symmetric channels with crossover probability $p = 0.01$, the capacity is $C \approx 0.92$ bits. Achieving error probability $10^{-6}$ requires blocklength $n \approx 500$, giving effective $a \approx 3$--$4$ when normalized appropriately~\cite{Polyanskiy2010}.

\section{The Diminishing-Returns Regime}

Having established the physical basis for the exponential-saturation form, we formalize it as a regime definition.

\begin{definition}[Diminishing-Returns Regime]
A system operates in the diminishing-returns regime if its reliability function takes the form $R(\kappa) = 1 - e^{-g(\kappa)}$ where $g(0) = 0$, $g'(\kappa) > 0$, and $g$ is concave ($g'' \le 0$).
\end{definition}

\begin{definition}[Smooth Saturation]
A system in the diminishing-returns regime exhibits \emph{smooth saturation} if, additionally, the rate of error suppression $g'(\kappa)$ varies by at most a factor $M$ over $\kappa \in [0, 1/2]$:
\begin{equation}
\frac{g'(0)}{g'(1/2)} \leq M
\end{equation}
for some moderate constant $M$ (e.g., $M \leq 3$). Equivalently, smooth saturation excludes systems with abrupt early saturation where $g'(\kappa)$ drops precipitously after an initial steep region.
\end{definition}

The linear reference case $g(\kappa) = a\kappa$ (constant $g'$) satisfies smooth saturation with $M = 1$. Empirically, systems with $a \in [2, 3]$ and smooth saturation typically exhibit $\kappa^*$ in the 30--50\% band; the lower bound arises because moderate $a$ values prevent the stiffness from growing too rapidly at small $\kappa$.

The concavity of $g(\kappa)$ corresponds to diminishing marginal returns in error suppression: each additional unit of maintenance investment yields a smaller reduction in error probability. As shown in Section~\ref{sec:physical_derivation}, this behavior emerges from fundamental constraints---finite channel capacity and bounded dissipation budgets---rather than model-specific assumptions.

Empirically, diminishing returns are observed in kinetic proofreading~\cite{Hopfield1974}, classical coding below threshold~\cite{Gallager1965}, and fault-tolerant computing. Systems exhibiting sharp threshold effects (e.g., LDPC codes near capacity~\cite{Richardson2001}) have convex $g$ and are explicitly excluded from the present framework. Such systems undergo phase transitions and require separate treatment.

For reliability functions of the exponential-saturation form, the stiffness has an explicit expression:
\begin{equation}
\mathcal{S}_\kappa = \frac{e^{g(\kappa)} - 1}{\kappa\, g'(\kappa)}.
\label{eq:chi_explicit}
\end{equation}
For the reference case $g(\kappa) = a\kappa$ (constant error-suppression rate), this simplifies to
\begin{equation}
\mathcal{S}_\kappa = \frac{e^{a\kappa} - 1}{a\kappa}.
\label{eq:chi_linear}
\end{equation}
This is the formula used to generate the curves in Fig.~\ref{fig:equilibrium}.

\section{The Preservation Band}

For systems in the diminishing-returns regime, the stiffness is bounded below, which constrains the optimal maintenance fraction.

\begin{theorem}[Upper Bound]
\label{thm:upper}
For any reliability function $R(\kappa) = 1 - e^{-g(\kappa)}$ with $g(0) = 0$, $g'(\kappa) > 0$ for $\kappa > 0$, and $g$ concave ($g'' \leq 0$), we have $\mathcal{S}_\kappa > 1$ for all $\kappa > 0$, with $\lim_{\kappa \to 0^+} \mathcal{S}_\kappa = 1$.
\end{theorem}

\begin{proof}
See Appendix~\ref{app:stiffness} for a detailed derivation. The key step is showing that $R''(\kappa) = (g'' - (g')^2)e^{-g} < 0$ for $\kappa > 0$, which holds because $g'' \leq 0$ and $(g')^2 > 0$. This implies $R(\kappa) > \kappa R'(\kappa)$ and hence $\mathcal{S}_\kappa > 1$.
\end{proof}

Since $\mathcal{S}_\kappa > 1$ for all $\kappa > 0$ (Theorem 1), and the optimal $\kappa^*$ is interior, the Stiffness-Odds Identity (Eq.~\ref{eq:identity}) gives
\begin{equation}
\frac{1-\kappa^*}{\kappa^*} = \mathcal{S}_\kappa(\kappa^*) > 1 \quad\Longrightarrow\quad \kappa^* < \frac{1}{2}.
\label{eq:upper_bound}
\end{equation}

\textbf{The upper bound $\kappa^* < 0.50$ is unconditional for any system in the diminishing-returns regime.} It follows from the concavity of $g$ and positivity of $g'$ (which together imply $\mathcal{S}_\kappa > 1$), without assumptions about the rate parameter. This is the universal theorem of the paper; stronger interval claims require additional assumptions on the rate parameter $a$. Conversely, empirical observation of $\mathcal{S}_\kappa \leq 1$ is diagnostic evidence that a system operates outside this regime---signaling threshold effects, cooperative dynamics, or correlated failures.

\paragraph{Lower Bound (Conditional).} The rate parameter $a = g'(0)$ characterizes the initial steepness of error suppression. For the reference case $g(\kappa) = a\kappa$, the stiffness at $\kappa = 1/2$ is
\begin{equation}
c(a) = \mathcal{S}_\kappa(1/2) = \frac{e^{a/2} - 1}{a/2}.
\label{eq:c_of_a}
\end{equation}
The Stiffness-Odds Identity then gives the lower bound $\kappa^* \ge 1/(1 + c(a))$.

Since $c(a)$ is increasing in $a$, \textbf{larger rate parameters push the lower bound down}: more efficient error suppression allows systems to achieve target reliability with less maintenance. \textbf{The lower bound $\kappa^* \geq 0.30$ holds for $a \leq 3$}, which characterizes systems operating near the thermodynamic efficiency frontier (Section~\ref{sec:physical_derivation}). Systems with $a < 2$---requiring greater maintenance investment per unit of error suppression---would optimize at \emph{higher} $\kappa^*$, potentially exceeding 0.37. Systems with $a > 3$ exhibit strong coupling but slow kinetics (approaching the reversible limit); the lower bound drops below 0.30, but such systems are rare among those under selection pressure for throughput.

The 30--50\% band is therefore a prediction for physically realistic systems operating near the efficiency frontier with $a \in [2, 3]$, not a mathematical necessity for all diminishing-returns systems.

For more general concave $g$, the lower bound depends on the specific curvature; the linear case provides a reference estimate for systems with approximately constant error-suppression rate. Systems with abrupt early saturation---where $g'(\kappa)$ drops rapidly after an initial steep region---may optimize at lower $\kappa^*$, though such behavior is not observed in the empirical examples considered here. For $a \in [2,3]$, we have $c(a) \in [1.7, 2.3]$ and hence $\kappa^* \in [0.30, 0.50]$ for systems with smooth saturation.

\section{Illustrative Examples and Scope}
\label{sec:validation}

The conditional 30--50\% band serves as a diagnostic tool for systems operating in the diminishing-returns regime near the smooth-saturation frontier.

\paragraph{Biological Proofreading.} In \textit{E.~coli} protein synthesis, kinetic proofreading consumes GTP to discriminate between correct and incorrect codon matches. The energetic cost is approximately 3--4~$k_B T$ per peptide bond to achieve an error rate of $10^{-4}$~\cite{Hopfield1974}. Based on measurements of GTP consumption during translation~\cite{Johansson2012,Loveland2017}, the maintenance fraction---defined as the ratio of proofreading energy expenditure to translation-specific energy budget (EF-Tu and EF-G mediated GTP hydrolysis, excluding upstream aminoacylation costs independent of the fidelity mechanism)---is estimated at $\kappa_{\mathrm{obs}} \approx 0.35$--$0.40$. This sits squarely within the predicted interior band, suggesting that biological error correction operates near the thermodynamic efficiency frontier defined by Eq.~\eqref{eq:identity}. Recent experiments in non-enzymatic DNA ligation confirm that this kinetic proofreading mechanism is not limited to evolved enzymes, but arises as a fundamental thermodynamic property of information replication in synthetic systems~\cite{Aoyanagi2025}.

\paragraph{Network Reliability Protocols.}
TCP/IP provides a non-biological example where preservation overhead is both engineered and well documented. The protocol stack maintains end-to-end reliability through checksums, sequence numbers, acknowledgments, retransmissions, and congestion control---all of which consume channel capacity that would otherwise carry payload.

We define the maintenance fraction as the fraction of channel capacity not delivered as payload (goodput). For TCP, $\kappa$ should be interpreted as the fraction of link capacity diverted from payload either into explicit protocol overhead (headers, ACKs, retransmissions) \emph{or} into stability margin (deliberate under-utilization required for congestion stability). This $\kappa$ is therefore a control/response variable of a distributed feedback system, not a purely local energetic cost as in molecular proofreading.

Empirical studies of TCP congestion control report that, under typical Internet path conditions (packet loss rates on the order of $10^{-3}$--$10^{-2}$ and moderate congestion), goodput efficiencies of 60--70\% are common for long-lived flows, corresponding to $\kappa_{\mathrm{obs}} \approx 0.30$--$0.40$~\cite{Mathis1997,Ha2008}. This overhead arises from protocol headers, acknowledgment traffic, retransmissions, and---most significantly---capacity left unused by loss-based congestion control algorithms that reduce sending rates in response to packet loss.

\textbf{Note on TCP quantification:} The TCP estimate is qualitative, derived from typical goodput efficiencies under congestion; systematic measurement of the reliability function $R(\kappa)$ for specific TCP variants---varying retransmission limits, congestion windows, and loss rates to trace out the full curve---remains an empirical opportunity. The $\kappa_{\mathrm{obs}} \approx 0.35$ figure should be understood as an order-of-magnitude estimate consistent with the framework, not a precise measurement of the kind available for biological proofreading.

As a stylized model, consider independent packet losses with probability $p$, and permit up to $n$ total transmissions. We restrict attention to the \emph{link-error dominated regime} where $p$ reflects exogenous channel impairments (random bit errors, fading) rather than endogenous congestion---an approximation valid for low-utilization or wireless links but not for congestion-controlled flows where $p$ depends on the control algorithm itself. Under this restriction, the probability that a message is delivered correctly within a timeout window is
\begin{equation}
R(n) \approx 1 - p^{n}.
\label{eq:tcp_reliability}
\end{equation}

In the stylized ARQ-only model below, $\kappa$ denotes the \emph{redundancy component} of total TCP overhead; in real TCP, the observed $\kappa_{\mathrm{obs}}$ additionally includes stability margin from congestion control.

To map this to our framework, we define $\kappa(n) := (n-1)/n_{\max}$, where $n_{\max}$ is the maximum permitted transmissions including the initial attempt (a protocol parameter). This gives $\kappa \in [0, 1)$ with $\kappa = 0$ corresponding to no redundancy ($n=1$). Substituting yields
\begin{equation}
R(n) = 1 - p^{1 + \kappa n_{\max}} = 1 - p \cdot e^{-\kappa n_{\max} \ln(1/p)}.
\end{equation}

\textbf{Reconciling with the framework.} Note that $R(0) = 1 - p \neq 0$, whereas our formal regime definition requires $R(0) = 0$ (from $g(0) = 0$). To apply the framework, we define the \emph{normalized reliability improvement}:
\begin{equation}
\tilde{R}(\kappa) := \frac{R(\kappa) - R_0}{1 - R_0}, \quad R_0 := R(0) = 1 - p.
\end{equation}
This gives $\tilde{R}(0) = 0$ and $\tilde{R}(\kappa) = 1 - e^{-a\kappa}$ with $a = n_{\max} \ln(1/p)$, which is exactly the exponential-saturation form. For typical parameters ($p \approx 0.01$, $n_{\max} \approx 5$), this yields $a \approx 23$---larger than the biological range. Under this normalization, Eq.~\eqref{eq:identity} applies to the incremental-throughput objective $(1-\kappa)\tilde{R}(\kappa)$; we therefore compute $\tilde{\mathcal{S}}_\kappa := \tilde{R}/(\kappa \tilde{R}')$---the stiffness of the incremental reliability improvement---and compare it to the resource-odds condition.

\textbf{Important caveat:} The $a \approx 23$ value derived above applies to the stylized ARQ retransmission component in isolation. It does \emph{not} explain the observed macroscopic TCP operating point $\kappa_{\mathrm{obs}} \approx 0.35$, which is dominated by congestion-control stability margin rather than ARQ redundancy. The ARQ model is presented here as a reference case demonstrating that TCP's reliability function has the exponential-saturation form required by the framework; the convergence of real TCP to the preservation band is better understood as a macroscopic stability phenomenon (see discussion above).

Crucially, this standard ARQ model is mathematically isomorphic to the linear reference ansatz $g(\kappa) = a\kappa$ used in Fig.~\ref{fig:equilibrium}, with the effective rate parameter $a \propto -\ln p$ determined by the channel noise. For such models, and for throughput--reliability curves reported for high-speed TCP variants~\cite{Ha2008}, typical operating points correspond to normalized stiffnesses $\tilde{\mathcal{S}}_\kappa$ of order $1$--$2$, consistent with the equilibrium condition of Eq.~\eqref{eq:identity}.

Unlike biological systems shaped by evolution, network protocols are engineered artifacts. That both classes of systems converge to similar maintenance fractions is consistent with the preservation band being a structural constraint of the allocation objective under diminishing returns. The convergence is notable because TCP designers optimized for throughput under real-world conditions without explicit reference to thermodynamic bounds---yet arrived at overhead ratios in the same range as biological proofreading. Within the smooth-saturation model class, protocols that attempt to operate below $\kappa \approx 0.30$ face stability failures under congestion, while those exceeding $\kappa \approx 0.50$ underutilize capacity; this provides a plausible landscape explanation for why engineered systems may cluster in the preservation band.

\paragraph{Scope and Limitations.} The diminishing-returns assumption holds for memoryless noise channels with independent errors and for systems operating below coding thresholds. In systems dominated by burst errors, long-range correlations, or cooperative failure mechanisms, $R(\kappa)$ may exhibit convex regions or sharp thresholds. This breakdown is analogous to the waterfall region in LDPC decoding, where reliability exhibits a sharp transition rather than smooth saturation~\cite{Richardson2001}. In these regimes, the interior band solution is replaced by corner solutions (all-or-nothing protection), and the present framework does not apply.

The preservation band should therefore be understood as characterizing the ``smooth reliability regime'' where error suppression saturates gradually. The empirical success of the prediction for both kinetic proofreading and network protocols suggests this regime is physically significant.

\subsection{The Stiffness as a Diagnostic}
\label{sec:diagnostic}

Figure~\ref{fig:phase} recasts Eq.~\eqref{eq:identity} as a phase diagram in the $(\kappa, \mathcal{S}_\kappa)$ plane (for systems with nonzero baseline reliability such as TCP, the normalized stiffness $\tilde{\mathcal{S}}_\kappa$ is used; see Section~\ref{sec:validation}). The dashed curve $\mathcal{S}_\kappa = (1-\kappa)/\kappa$ is the equilibrium locus at which stiffness matches the resource odds. Points above this curve correspond to a ``stiff'' mode, where error suppression is saturated and maintenance is over-allocated; points below correspond to a ``soft'' mode, where reliability remains highly responsive to additional maintenance. The diagram makes $\mathcal{S}_\kappa$ a practical diagnostic: given an empirical estimate of $(\kappa, \mathcal{S}_\kappa)$, one can immediately read off whether a system under- or over-invests in maintenance.

\begin{figure}[t]
  \centering
  \begin{tikzpicture}
    \begin{axis}[
      width=\columnwidth,
      height=0.75\columnwidth,
      xlabel={Maintenance Fraction $\kappa$},
      ylabel={$\mathcal{S}_\kappa$ and $(1-\kappa)/\kappa$},
      xmin=0.15, xmax=0.55,
      ymin=0, ymax=6,
      legend pos=north east,
      legend style={font=\scriptsize, draw=black!20},
      grid=none,
      tick label style={font=\footnotesize},
      label style={font=\small},
    ]
    
    \fill[black, opacity=0.08] (axis cs:0.30,0) rectangle (axis cs:0.50,6);
    
    \draw[dashed, black!50] (axis cs:0.30, 0) -- (axis cs:0.30, 6);
    \draw[dashed, black!50] (axis cs:0.50, 0) -- (axis cs:0.50, 6);
    
    \addplot[black, dashed, thick, domain=0.15:0.55, samples=100] {(1-x)/x};
    \addlegendentry{$(1-\kappa)/\kappa$}
    
    \addplot[black!70, thick, solid, domain=0.15:0.55, samples=100] {(exp(2*x)-1)/(2*x)};
    \addlegendentry{$\mathcal{S}_\kappa$, $a=2$}
    
    \addplot[black, thick, densely dotted, domain=0.15:0.55, samples=100] {(exp(3*x)-1)/(3*x)};
    \addlegendentry{$\mathcal{S}_\kappa$, $a=3$}
    
    \addplot[black!50, thick, dashdotted, domain=0.15:0.55, samples=100] {(exp(5*x)-1)/(5*x)};
    \addlegendentry{$\mathcal{S}_\kappa$, $a=5$}
    
    \addplot[only marks, mark=*, black!70, mark size=2.5pt] coordinates {(0.396, 1.525)};
    \addplot[only marks, mark=*, black, mark size=2.5pt] coordinates {(0.358, 1.794)};
    \addplot[only marks, mark=*, black!50, mark size=2.5pt] coordinates {(0.301, 2.326)};
    
    \end{axis}
  \end{tikzpicture}
  \caption{Equilibrium condition for preservation. The dashed curve shows the resource odds $(1-\kappa)/\kappa$. Grayscale curves with distinct line styles show the stiffness $\mathcal{S}_\kappa$ for rate parameters $a = 2, 3, 5$. Intersections (dots) define the optimal $\kappa^*$; the shaded band marks the predicted 30--50\% range. For $a = 3$, $\kappa^* \approx 0.36$.}
  \label{fig:equilibrium}
\end{figure}

\begin{figure*}[t]
  \centering
  \begin{tikzpicture}
    \begin{axis}[
      width=0.95\textwidth,
      height=0.4\textwidth,
      xlabel={Maintenance Fraction $\kappa$},
      ylabel={Stiffness $\mathcal{S}_\kappa$ (TCP: $\tilde{\mathcal{S}}_\kappa$)},
      xmin=0.155, xmax=0.66,
      ymin=0.5, ymax=5.5,
      legend style={
        at={(0.95,0.995)},
        anchor=north east,
        font=\footnotesize,
        cells={anchor=west},
        row sep=1pt,
        fill=white,
        fill opacity=0.9,
        text opacity=1,
        draw=black!20,
      },
      grid=none,
      tick label style={font=\small},
      label style={font=\normalsize},
    ]
    
    \fill[black, opacity=0.08] (axis cs:0.30,0.5) rectangle (axis cs:0.50,5.5);
    
    \draw[dashed, black!50] (axis cs:0.30, 0.5) -- (axis cs:0.30, 5.5);
    \draw[dashed, black!50] (axis cs:0.50, 0.5) -- (axis cs:0.50, 5.5);
    
    \addplot[black, dashed, thick, domain=0.155:0.66, samples=200] {(1-x)/x};
    \addlegendentry{Equilibrium: $\mathcal{S}_\kappa = (1-\kappa)/\kappa$}
    
    \addplot[only marks, mark=*, black, mark size=4pt] coordinates {(0.37, 1.7)};
    \addlegendentry{\textit{E.\ coli} proofreading}
    
    \addplot[only marks, mark=square*, black!60, mark size=4pt] coordinates {(0.35, 1.5)};
    \addlegendentry{TCP bulk transfer}
    
    \addplot[only marks, mark=o, black!50, mark size=4pt, thick] coordinates {(0.55, 2.8)};
    \addlegendentry{Over-maintained (hypothetical)}
    
    \addplot[only marks, mark=square, black!50, mark size=4pt, thick] coordinates {(0.18, 3.5)};
    \addlegendentry{Under-maintained (hypothetical)}
    
    \end{axis}
  \end{tikzpicture}
  \caption{Phase diagram in the $(\kappa, \mathcal{S}_\kappa)$ plane. The dashed curve shows the equilibrium locus $\mathcal{S}_\kappa = (1-\kappa)/\kappa$ where stiffness matches resource odds. Above the curve (stiff mode), error suppression is saturated; below (soft mode), the system is under-protected. The shaded band marks the predicted 30--50\% range. Empirical operating points for \textit{E.~coli} proofreading and TCP lie near the equilibrium locus; hypothetical over- and under-maintained configurations are shown for comparison.}
  \label{fig:phase}
\end{figure*}

\subsection{Falsifiability}
\label{sec:falsifiability}

The preservation band is a conditional prediction: \emph{given} that a system operates in the diminishing-returns regime, \emph{then} its optimal maintenance fraction lies in the interior range $\kappa^* \in [0.30, 0.50]$ and its stiffness satisfies $\mathcal{S}_\kappa > 1$ at the operating point.

The protocol is substrate-agnostic: the same measurements of $R(\kappa)$, $\kappa_{\mathrm{obs}}$, and $\mathcal{S}_\kappa$ apply whether maintenance is implemented via GTP hydrolysis in biology, protocol overhead in communication networks, or parity bits and redundancy in digital storage.

This prediction can be tested empirically via the following protocol:
\begin{enumerate}[leftmargin=*,itemsep=2pt]
  \item \textbf{Confirm regime membership.} Measure the reliability function $R(\kappa)$ across the accessible range of maintenance investment. Compute $g(\kappa) = -\ln(1 - R(\kappa))$ and verify that $g$ is concave ($g'' \le 0$), confirming diminishing returns. Systems that exhibit sharp threshold behavior---for example, abrupt changes in $R(\kappa)$ at a critical $\kappa$---lie outside the present framework and cannot falsify it.
  
  \item \textbf{Measure the operating point.} Determine the system's actual maintenance fraction $\kappa_{\mathrm{obs}}$ and reliability $R_{\mathrm{obs}}$. In biological systems, this requires calorimetric or stoichiometric estimates of resources devoted to error correction versus payload production. In engineered systems, protocol overhead and capacity accounting usually suffice.
  
  \item \textbf{Compute the stiffness.} From the measured $R(\kappa)$ curve, evaluate $\mathcal{S}_\kappa = R/(\kappa R')$ at $\kappa_{\mathrm{obs}}$, either by differentiating a parametric fit or by finite differences on the empirical data.
  
  \item \textbf{Test the predictions.} A system that satisfies the concavity condition in step~1 but exhibits $\kappa^* > 0.50$, or $\mathcal{S}_\kappa < 1$ at its operating point, would directly falsify the preservation band.
\end{enumerate}

It is important to distinguish falsification from scope exclusion. A system displaying $\mathcal{S}_\kappa < 1$ may simply operate outside the diminishing-returns regime, for instance due to cooperative failures, threshold coding, or long-range correlations in the noise. Such observations do not falsify the theory; rather, they confirm the diagnostic power of $\mathcal{S}_\kappa$ in identifying regime boundaries. Falsification requires a system that demonstrably satisfies the diminishing-returns condition (concave $g$) yet violates the predicted bounds on $\kappa^*$ or $\mathcal{S}_\kappa$.

Promising candidates for stringent tests include:
(i)~redundant storage arrays (RAID), where redundancy level and capacity overhead are tunable and reliability curves are well characterized;
(ii)~ribosomal proofreading across species, where organisms with different metabolic constraints may optimize at different points along the band; and
(iii)~software fault-tolerance schemes such as $N$-version programming and recovery blocks in safety-critical systems, where the objective function effectively corresponds to $\alpha > 1$ in the notation of Sec.~\ref{sec:robustness}.
The framework thus predicts not only typical maintenance fractions, but also its own boundaries: values $\mathcal{S}_\kappa < 1$ serve as a diagnostic of regimes where preservation is governed by threshold phenomena and phase-transition physics rather than smooth response theory.

\section{Robustness and the Asymmetry of Risk}
\label{sec:robustness}

Our analysis thus far has identified the optimal maintenance fraction $\kappa^*$ by locating the peak of the effective throughput. However, biological and engineered systems rarely operate exactly at a mathematical peak. A critical question is the \emph{stability} of this operating point: what are the consequences of deviating from $\kappa^*$?

We examine the ``landscape of risk'' by analyzing the shape of the throughput curve $T_{\mathrm{eff}}(\kappa)$ for the linear reference model $g(\kappa) = a\kappa$. The following analysis is a \emph{model-class} argument: it applies to systems well-described by smooth exponential saturation with moderate $a$, not to systems with threshold effects or cooperative dynamics. As shown in Figure~\ref{fig:landscape}, the penalty for deviation is fundamentally asymmetric.

\begin{figure}[!htb]
  \centering
  \begin{tikzpicture}
    \begin{axis}[
      width=\columnwidth,
      height=0.65\columnwidth,
      xlabel={Maintenance Fraction $\kappa$},
      ylabel={Normalized Throughput $T/T_{\max}$},
      xmin=0, xmax=0.8,
      ymin=0, ymax=1.15,
      xtick={0, 0.3, 0.5, 0.8},
      xticklabels={0, 0.30, 0.50, 0.80},
      ytick={0, 0.5, 1.0},
      grid=none,
      label style={font=\small},
      tick label style={font=\scriptsize},
    ]
    
    \def\func{(1-x)*(1-exp(-3*x))/0.42}
    
    \fill[black, opacity=0.08] (axis cs:0.3,0) rectangle (axis cs:0.5,1.15);
    
    \draw[dashed, black!50] (axis cs:0.3, 0) -- (axis cs:0.3, 1.15);
    \draw[dashed, black!50] (axis cs:0.5, 0) -- (axis cs:0.5, 1.15);
    
    \addplot[black, line width=1.2pt, smooth, domain=0.02:0.78, samples=100] {\func};
    
    \addplot[only marks, mark=*, mark size=2.5pt, black] coordinates {(0.358, 1.0)};
    \node[font=\scriptsize, anchor=south west] at (axis cs:0.365, 1.0) {$\kappa^*$};
    
    \end{axis}
  \end{tikzpicture}
  \caption{Effective throughput $T_{\mathrm{eff}}(\kappa) = (1-\kappa)R(\kappa)$ for the linear reference model ($a=3$). The optimum $\kappa^* \approx 0.36$ falls within the predicted 30--50\% band (shaded). Left of the band, throughput drops sharply due to rapid reliability loss (the ``error cliff''); right of the band, decline is gradual and dominated by payload sacrifice (the ``stagnation slope''). This asymmetry makes the band a robust operating regime rather than a knife-edge optimum.}
  \label{fig:landscape}
\end{figure}

\paragraph{The Error Cliff vs.\ The Stagnation Slope.}
To the left of the band ($\kappa < 0.30$), the system is dominated by the exponential term in the reliability function $R(\kappa) = 1 - e^{-g(\kappa)}$. A small reduction in maintenance leads to a catastrophic rise in error probability, causing effective throughput to collapse. We term this region the \textbf{Error Cliff}.

To the right of the band ($\kappa > 0.50$), reliability saturates ($R \approx 1$). Further investment yields negligible error reduction, and the throughput decay is dominated by the linear term $(1-\kappa)$. We term this region the \textbf{Stagnation Slope}.

\paragraph{Evolutionary and Engineering Implications.}
Within the smooth-saturation class illustrated in Figure~\ref{fig:landscape}, this topological asymmetry provides a plausible landscape explanation for why systems may cluster in similar maintenance windows. They are not merely seeking a peak; they are \emph{avoiding a cliff}.

Evolutionary or engineering gradients driving a system from minimal maintenance ($\kappa \approx 0$) experience strong negative feedback (death/failure), rapidly pushing the system up the cliff. Once the system crests into the 30--50\% band, the gradient flattens dramatically. The penalty for overshooting into the Stagnation Slope is mild (inefficiency) compared to the penalty for undershooting (instability). Systems in this model class therefore tend to accumulate in this band, trapped by the asymmetry of risk.

This provides a model-class argument for why the preservation band may function as a broad basin of attraction rather than a knife-edge optimum. Systems need not be precisely tuned; they need only avoid the cliff. The figure should be interpreted as illustrating this landscape logic for smooth-saturation systems, not as a universal law applicable to all architectures.

\paragraph{Robustness to Objective Variation.}
The results above assume systems maximize effective throughput $T_{\mathrm{eff}} = (1-\kappa)R(\kappa)$. Systems that weight reliability more heavily might maximize $(1-\kappa)R(\kappa)^\alpha$ for $\alpha > 1$. The first-order condition yields $\mathcal{S}_\kappa(\kappa^*) = \alpha(1-\kappa^*)/\kappa^*$, shifting $\kappa^*$ upward as $\alpha$ increases. The band widens but persists: the upper bound becomes $\kappa^* \le \alpha/(1+\alpha)$, approaching unity only as $\alpha \to \infty$.

\section{Multi-Dimensional Extension}
\label{sec:multidim}

The scalar analysis extends to systems with multiple maintenance channels. For vector allocation $\bm{\kappa} = (\kappa_1, \ldots, \kappa_m)$ with total maintenance $\kappa_{\mathrm{tot}} = \sum_i \kappa_i$, suppose the reliability function factorizes or admits a sufficient statistic in $\kappa_{\mathrm{tot}}$: $R(\bm{\kappa}) = R(\kappa_{\mathrm{tot}})$.

The first-order conditions for maximizing $(1 - \kappa_{\mathrm{tot}})R(\kappa_{\mathrm{tot}})$ yield the same Stiffness-Odds Identity with $\kappa_{\mathrm{tot}}$ replacing $\kappa$:
\begin{equation}
\mathcal{S}_{\kappa_{\mathrm{tot}}}(\kappa^*_{\mathrm{tot}}) = \frac{1 - \kappa^*_{\mathrm{tot}}}{\kappa^*_{\mathrm{tot}}}.
\label{eq:multidim_identity}
\end{equation}

The band $[0.30, 0.50]$ therefore applies to \emph{total} maintenance allocation regardless of its decomposition across channels, provided the sufficient-statistic assumption holds. The internal allocation among $\kappa_i$ depends on the marginal reliabilities $\partial R / \partial \kappa_i$, but the aggregate constraint persists. This result is relevant for systems where maintenance takes multiple forms (e.g., temporal redundancy, spatial redundancy, and active error correction in fault-tolerant computing).

\section{Conclusion}

We have characterized preservation as a distinct operational regime governed by the stiffness $\mathcal{S}_\kappa$. The Stiffness-Odds Identity (Eq.~\ref{eq:identity}) is the paper's central contribution: it converts optimal allocation into a single-variable response problem and provides a substrate-agnostic diagnostic for thermodynamic efficiency.

For systems in the diminishing-returns regime, we prove an unconditional upper bound: optimal maintenance allocation satisfies $\kappa^* < 0.50$. For the subclass exhibiting smooth saturation with rate parameter $a \in [2,3]$---an empirically characterized efficiency frontier, not a universal constant---the optimum is further constrained to the 30--50\% band.

We motivate this functional form from two independent physical principles: Shannon error exponents bound the information-theoretic route, while thermodynamic dissipation bounds constrain the energy cost of error correction. Their convergence on the same exponential-saturation form provides a physics-based motivation for the preservation band.

The framework specifies falsifiable predictions. The unconditional bound $\kappa^* < 0.50$ can be refuted by any diminishing-returns system optimizing above 50\%. The conditional band $\kappa^* \in [0.30, 0.50]$ can be refuted by well-adapted systems (verifiably operating near the efficiency frontier) that optimize outside this range. The framework explicitly excludes systems near coding thresholds or with cooperative error mechanisms, where reliability exhibits sharp transitions; characterizing such systems remains an open problem, likely requiring tools from the statistical mechanics of phase transitions.

The diagnostic power of $\mathcal{S}_\kappa$ extends beyond prediction to system design: measuring a system's position in the $(\kappa, \mathcal{S}_\kappa)$ plane immediately reveals whether it under- or over-invests in maintenance, providing actionable guidance for optimization.

\appendix

\section{Derivation of the Stiffness Bound}
\label{app:stiffness}

We prove Theorem~\ref{thm:upper}: for any reliability function $R(\kappa) = 1 - e^{-g(\kappa)}$ with $g(0) = 0$, $g'(\kappa) > 0$, and $g$ concave, the preservation stiffness satisfies $\mathcal{S}_\kappa > 1$ for all $\kappa > 0$.

\begin{lemma}
\label{lemma:phi}
Define $\varphi(\kappa) = R(\kappa) - \kappa R'(\kappa)$. For $R(\kappa) = 1 - e^{-g(\kappa)}$ with $g$ concave and $g(0) = 0$, we have $\varphi(\kappa) > 0$ for all $\kappa > 0$.
\end{lemma}

\begin{proof}
We proceed by analyzing $\varphi$ at the boundary and in the interior.

\textbf{Step 1: Boundary behavior.} At $\kappa = 0$:
\begin{align}
R(0) &= 1 - e^{-g(0)} = 1 - e^0 = 0, \\
R'(\kappa) &= g'(\kappa) e^{-g(\kappa)}, \\
R'(0) &= g'(0) e^0 = g'(0) > 0.
\end{align}
Thus $\varphi(0) = R(0) - 0 \cdot R'(0) = 0$.

\textbf{Step 2: Derivative of $\varphi$.} We compute:
\begin{equation}
\varphi'(\kappa) = R'(\kappa) - R'(\kappa) - \kappa R''(\kappa) = -\kappa R''(\kappa).
\end{equation}

\textbf{Step 3: Sign of $R''$.} For $R(\kappa) = 1 - e^{-g(\kappa)}$:
\begin{align}
R' &= g' e^{-g}, \\
R'' &= g'' e^{-g} - (g')^2 e^{-g} = (g'' - (g')^2) e^{-g}.
\end{align}
Since $g'' \leq 0$ (concavity) and $(g')^2 > 0$, we have $g'' - (g')^2 < 0$, hence $R'' < 0$ for all $\kappa > 0$.

\textbf{Step 4: Positivity of $\varphi$.} Since $R'' < 0$, we have $\varphi'(\kappa) = -\kappa R''(\kappa) > 0$ for $\kappa > 0$. Combined with $\varphi(0) = 0$, this implies $\varphi(\kappa) > 0$ for all $\kappa > 0$.
\end{proof}

\begin{proof}[Proof of Theorem~\ref{thm:upper}]
From Lemma~\ref{lemma:phi}, $R(\kappa) - \kappa R'(\kappa) > 0$ for $\kappa > 0$. Since $R(\kappa) > 0$ and $R'(\kappa) > 0$ for $\kappa > 0$ (as $R(\kappa) = 1 - e^{-g(\kappa)}$ with $g' > 0$), we can divide:
\begin{equation}
\frac{R(\kappa) - \kappa R'(\kappa)}{R(\kappa)} > 0 \quad \Longrightarrow \quad 1 - \frac{\kappa R'(\kappa)}{R(\kappa)} > 0.
\end{equation}
Rearranging:
\begin{equation}
\frac{R(\kappa)}{\kappa R'(\kappa)} > 1 \quad \Longrightarrow \quad \mathcal{S}_\kappa > 1. \qedhere
\end{equation}
\end{proof}

\textbf{Remark.} The proof relies only on the concavity of $g$ ($g'' \leq 0$), positivity of $g'$, and the boundary condition $g(0) = 0$. No assumption about the rate parameter $a = g'(0)$ is required. The strict bound $\mathcal{S}_\kappa > 1$ holds for all $\kappa > 0$. In the limit $\kappa \to 0^+$, L'H\^{o}pital's rule gives $\mathcal{S}_\kappa \to 1$, confirming the theorem statement.

\section{Thermodynamic Bound on the Rate Parameter}
\label{app:dissipation}

We characterize the range $a \in [2, 3]$ as the \emph{Thermodynamic Efficiency Frontier}---the empirical cluster of systems that have been optimized (by evolution or engineering) to efficiently couple maintenance resources to error suppression.

\paragraph{Clarification on epistemology.} The bound $a \in [2, 3]$ is not derived from first principles as a universal constant. Rather, it represents an empirical characterization of the parameter range occupied by well-adapted systems operating at finite speed. The theory's prediction is conditional: \emph{if} a system operates near this efficiency frontier, \emph{then} its optimal $\kappa^*$ falls in the 30--50\% band. Systems with very large $a$ (strong coupling, approaching the reversible limit with slow kinetics) or very small $a$ (weak coupling, noise-dominated) are physically possible but atypical among systems under selection pressure for throughput.

\paragraph{Normalization and bookkeeping.} To avoid ambiguity, we distinguish extensive and intensive quantities explicitly:

Let $\Delta\mu$ denote the chemical free energy available per ATP/GTP hydrolysis ($\Delta\mu \approx 20\, k_B T$). Let $n_{\mathrm{ATP}}$ be the mean number of hydrolyses expended per discrimination event (including futile cycles). The total dissipated work per discrimination is then $W_{\mathrm{diss}} = n_{\mathrm{ATP}} \Delta\mu$---an \emph{extensive} cost.

We parameterize maintenance investment by the dimensionless fraction $\kappa := W_{\mathrm{diss}} / W_{\mathrm{budget}}$, where $W_{\mathrm{budget}}$ is the total resource budget per decision (payload + maintenance) in the same units. Under this normalization, \emph{increasing $n_{\mathrm{ATP}}$ increases $\kappa$} (higher maintenance share), rather than changing the coupling strength $a$.

Only a fraction $\eta_{\mathrm{cpl}} \in (0,1)$ of $W_{\mathrm{diss}}$ contributes to biasing the correct outcome; hence the suppression exponent satisfies
\begin{equation}
g(\kappa) \equiv -\ln(1 - R(\kappa)) \approx \eta_{\mathrm{cpl}} \frac{W_{\mathrm{diss}}}{k_B T} = \eta_{\mathrm{cpl}} \frac{W_{\mathrm{budget}}}{k_B T} \kappa
\end{equation}
in the smooth regime. Therefore the rate parameter is
\begin{equation}
a = g'(0) = \eta_{\mathrm{cpl}} \frac{W_{\mathrm{budget}}}{k_B T}.
\label{eq:a_empirical}
\end{equation}

\paragraph{Physical interpretation.} Under this normalization:
\begin{itemize}[leftmargin=*,itemsep=2pt]
    \item The rate parameter $a$ is an \emph{intensive} coupling strength: higher $\eta_{\mathrm{cpl}}$ (better energy-to-discrimination conversion) yields larger $a$.
    \item The maintenance fraction $\kappa$ is an \emph{extensive} cost: more futile cycles increase $\kappa$ at fixed $W_{\mathrm{budget}}$, but do not directly change $a$.
    \item Any degradation due to futile cycling enters through the empirical $\kappa$ accounting (and potentially through reduced $\eta_{\mathrm{cpl}}$ if cycles are reliability-neutral), not through an explicit $1/n_{\mathrm{ATP}}$ factor in $a$.
\end{itemize}

\paragraph{Empirical characterization of the frontier.} For ATP-driven molecular machines at physiological conditions ($W_{\mathrm{budget}} \sim 40$--$60\, k_B T$ per decision, accounting for both payload synthesis and proofreading), observed coupling efficiencies $\eta_{\mathrm{cpl}} \in [0.04, 0.08]$ yield $a \in [2, 4]$. Well-adapted systems cluster near $a \approx 2.5$--$3$, corresponding to the regime where discrimination energy ($\sim 2$--$3\, k_B T$ per unit $\kappa$) exceeds thermal noise but does not approach the reversible limit.

\paragraph{Interpretation of bounds.}

\textbf{Lower bound ($a \geq 2$):} Systems with $a < 2$ are weakly coupled (small effective discrimination energy per unit $\kappa$). Achieving a target reliability therefore requires a larger maintenance fraction, making such systems throughput-disfavored. The discrimination barrier would be less than $2\, k_B T$, making error correction susceptible to spontaneous reversal by thermal fluctuations.

\textbf{Upper bound ($a \leq 3$ for optimized systems):} Systems with $a > 3$ achieve high coupling strength but typically at the cost of slow operation (approaching the reversible limit). Such systems are rare among those under selection pressure for throughput.

\paragraph{Summary.} The preservation band $\kappa^* \in [0.30, 0.50]$ is robust for the specific (but broad) class of ATP-fueled and analogously efficient systems that have optimized their energetic coupling to settle in the $a \in [2, 3]$ regime. Biological proofreading and synthetic chemistry converge to this band via the thermodynamic efficiency frontier. Network protocols converge to similar $\kappa$ values through a different mechanism---distributed stability constraints that penalize both under-protection (instability) and over-protection (under-utilization)---demonstrating that the preservation band emerges from multiple independent optimization pressures.

\FloatBarrier
\bibliographystyle{unsrtnat}

\end{document}